\documentclass{article}
\usepackage[round]{natbib}
\usepackage{amssymb, mathrsfs, amsthm, amsmath}
\newtheorem{definition}{Definition}
\newtheorem{proposition}{Proposition}
\newtheorem{lemma}{Lemma}
\newtheorem{theorem}{Theorem}

\title{Core equivalence with large agents}
\author{Aubrey Clark \footnote{I thank Rabee Tourky, Andrew McLennan, Eric Maskin, Oliver Hart, and Mihai Manea for comments. Email: aubs.bc@gmail.com}}

\begin{document}

\maketitle

\begin{abstract}
This paper studies the relationship between core and competitive equilibira in economies that consist of a continuum of agents and some large agents. We construct a class of these economies in which the core and competitive allocations do not coincide.
\end{abstract}

\section{Introduction}

Competitive equilibrium assumes price taking behavior. An old line of research dating back to \cite{Edgeworth81} has sought to provide a foundation for this assumption by constructing models in which the ability to interact strategically necessarily leads to a competitive outcome (e.g. \cite{DebreuScarf63}, \cite{Aumann64}, \cite{Gale86}).

\cite{DebreuScarf63} showed that every core allocation is eventually competitive as the size of the economy increases through replication, and \cite{Aumann64} showed that if one models the set of agents as an atomless measure space, then the core and competitive allocations coincide. Subsequent literature (e.g. \cite{Shitovitz73}, \cite{GabszewiczMertens1971}) has studied settings in which core allocations are necessarily competitive even though the underlying set of agents contains some atoms (i.e., large agents, or groups of agents that act collusively).

\cite{Shitovitz73} showed that core and competitive allocations coincide when there are at least two atoms and all atoms have the same convex preferences and initial endowment. \cite{GabszewiczMertens1971} showed that core and competitive allocations coincide when for each atom there is a sufficiently large group of consumers with the same preferences and initial endowment.

We add to these literatures by constructing economies in which a large agent exploits the rest of the economy in such a way that the resulting core allocation cannot be decentralized with prices.

\section{Model}

\paragraph{The economy} Our model of an economy consists of a set of consumers each of whom possess a commodity bundle and hold a preference relation over the set of all commodity bundles. We call the space $ \mathbb{R} ^ l $ the \emph{commodity space}, where $ l $ refers to the number of commodities. A commodity bundle is an element of the commodity space. A preference relation, denoted $ \succeq $, is a linear ordering \footnote{A complete, transitive, and reflexive binary relation} over the commodity space and induces a strict preference relation (denoted $ \succ $) in the following way: If $ x $ and $ y $ are arbitrary commodity bundles, then $ x \succ y $ if and only if $ x \succeq y $ and not $ y \succeq x $. The set of all such preference relations is $ \mathscr{P} $. The relation $ x \succeq y $ is read “$ x $ is preferred or indifferent to $ y $.”

The set of consumers in our economy is $ T $. Let $ (T, \mathscr{T} , \mu) $ be a measure space. An \emph{allocation} is an integrable function from $ T $ into the non-negative orthont of the commodity space. There is a fixed allocation, denoted $ \omega $, that is called the \emph{initial allocation}; $ \omega(t) $ is called the \emph{initial endowmenti} of consumer $ t $ . An allocation $ x $ is said to be \emph{ feasible } if $ \int x = \int \omega $.

Each consumer, represented by an element of $ T $, has a preference relation defined by a \emph{preference function} $ \succeq: T \rightarrow \mathscr{P} $. We say that $ \succeq: T \rightarrow \mathscr{P} $ is a \emph{measurable preference function} if, given arbitrary allocations $ x $ and $ y $, the set $ \{t \in T : x(t) \succ_t y(t)\} $ is measurable.

\begin{definition}
An economy, denoted $ \mathscr{E} $, is a triad $ ((T, \mathscr{T}, \mu), \omega, \succeq) $. It consists of a positive and finite measure space $ (T, \mathscr{T}, \mu) $, an initial allocation $ \omega: T \rightarrow \mathbb{R}^l_+$, and a measurable preference function $ \succeq : T \rightarrow \mathscr{P}$.
\end{definition}

Let $(T, \mathscr{T}, \mu) $ be a measure space. An \emph{atom} is a non-null element of the $\sigma$-algebra with the property that any subset also belonging to the $\sigma$-algebra is of equal or zero measure. That is, $ A \in \mathscr{T}$ is an atom if for all $B \subseteq A$, where $B \in \mathscr{T}$, we have $\mu(B) = \mu(A)$, or otherwise $\mu(B) = 0$. A measure space with atoms is called \emph{atomic}. A measure space without atoms is called \emph{atomless}.

\paragraph{Core and competitive equilibrium} Fix an economy $\mathscr{E}$. A \emph{price system}, denoted $p$, is a vector in $\mathbb{R}^l_+$. 

\begin{definition}
A competitive equilibrium $ (p, x) $ is a price system $ p $ and an allocation $ x $ such that $x$ is feasible, and for $\mu$-almost all $t$ in $T$
\begin{enumerate}
\item $p \cdot x(t) \leq p \cdot \omega(t)$, and
\item $y \succ_t x(t)$ implies $p \cdot y > p \cdot \omega(t)$.
\end{enumerate}
\end{definition}

If $ (p, x)$ is a competitive allocation, then $p$ is called an \emph{equilibrium price system} and $x$ is called a \emph{competitive allocation} The set of competitive allocations is denoted $W(\mathscr{E})$.

The other equilibrium concept we consider is the core. A \emph{coalition} is any $S$ in $\mathscr{T}$ with positive measure. We say that an allocation $x$ is \emph{blocked} by a coalition $S$ if there exists an allocation $y$ such that $\int_S y = \int_S \omega$, and $y(t) \succ_t x(t)$ for all $t$ in $S$.

\begin{definition}
The core of the economy $\mathscr{E}$ is the set of feasible allocations that are not blocked by any coalition.
\end{definition}

The core is denoted $C(\mathscr{E})$

\paragraph{Non-trivial linear economies} We call the economy \emph{linear} if every consumer's preference relation can be represented by a linear function $u_t: \mathbb{R}^l_+ \rightarrow \mathbb{R}$ defined by $u_t(x) = a_t \cdot x$, where $a_t$ is a vector in $\mathbb{R}^l_{++}$ such that $\sum_{i = 1} ^ l a_{ti} = 1$. An allocation for an economy is called \emph{Pareto optimal} if the coalition consisting of all consumers does not block the allocation. We call an economy \emph{non-trivial} if the initial allocation is not Pareto optimal.

Our first result provides a motivation for considering non-trivial economies. Recall that we want to characterise economies for which the core does not coincide with the set of competitive alklocations. Under fairly weak conditions on preferences, if the initial endowment is Pareto optimal then the core coincides with competitive allocations. Thus we exclude these economies from the analysis.

\begin{proposition}
Let $\mathscr{E}$ be an economy satisfying

\begin{enumerate}

  \item Monotonocity: let $x$ and $y$ be arbitrary commodity bunbles. Then\footnote{Let $x$ and $y$ be arbitrary points in $\mathbb{R}^l$. We write $x > y$ to mean that $x \neq y$ and $x_i \geq y_i$ for all $i = 1,2,.\cdots,l$} $x > y$ implies $x \succ_t y$
  \item Continuity: let $x$ be an arbitrary commodity bundle. The sets $\{z \in \mathbb{R}^l_+ : x \succ_t z\}$ and $\{z \in \mathbb{R}^l_+ : z \succ_t x\}$ are open
  \item Each component of $ \int_T \omega $ is positive
  \item The economy $\mathscr{E}$ has finitely many atoms and each atom has convex preferences    
\end{enumerate}

If $\omega$ is Pareto optimal, then $C(\mathscr{E}) = W(\mathscr{E})$.
\end{proposition}

\begin{proof}
Consider the allocation $f: T \rightarrow \mathbb{R}^l_+$ assigning almost all agents a commodity bundle preferred or indifferent to their initial endowment, and a non-null set of agents a commodity bundle preferred to their initial endowment. Let $F$ be the set of integrals of all such allocations. Let $I$ be the set of integrals over the atomless part of the economy of all such allocations. Let $G$ be the set of integrals over the atomic part of the economy of all such allocations. We have that $F = I + G$. By Lyapunov's convexity theorem $I$ is convex. By the convexity of each atom's preferences and the assumption that there are finmitely many atoms $G$ is convex. Thus $F$ is convex as the sum of two convex sets. The set $F$ is disjoint from $\int \omega$ because $\omega$ is Pareto optimal. By Minkowski's separating hyperplane theorem there exists a price sytstem $p$ such that $x \in F$ implies $p \cdot x \geq p \cdot \int \omega$. We claim $(p, \omega)$ is a competitive equilibrium. Suppose not. Then there exists a non-null set of agents $S$ and a commodity bundle $y(t)$ for each agent in $S$ such that $y(t) \succ_t \omega_t$ and $p \cdot y(t) \leq p \cdot \omega(t)$. Define $z: T \rightarrow \mathbb{R}^l$ by

\[z(t) = 
\begin{cases} 
      \omega(t) & \text{if } t \in T \setminus S \\
      y(t) & \text{if } t \in S.
\end{cases}
\]

We have $p \cdot \int z \geq p \cdot \int \omega$. If this holds with equality, then for almost all agents in $S$ we have $p \cdot y(t) = p \cdot \omega(t)$. By continuity of preferences, $y(t) - (\epsilon_t, \cdots, \epsilon_t) \succ_t \omega(t)$ for sufficiently small $\epsilon_t > 0$. Define 

\[z(t) = 
\begin{cases} 
      \omega(t) & \text{if } t \in T \setminus S \\
      y(t) -  (\epsilon_t, \cdots, \epsilon_t)& \text{if } t \in S.
\end{cases}
\]

By monotonicity and linearity of the integral we have $p \cdot \int \bar z < p \cdot \int \omega$, contradicting that $\int \bar z$ belongs to $F$. Thus

\[
p \cdot \int z > p \cdot \int \omega
\]

which implies $p \cdot \int_S y > p \cdot \int_S \omega$, contradicting $p\cdot y(t) \leq p \cdot \omega(t)$ for all $t \in S$. Therefore $(p, \omega)$ is a competitive equilibrium.

We now apply an argument of \cite{Aumann64} to show that every core allocation is a competitive allocation. For any feasible allocation $x$ let $\Gamma^x_t$ be the set of all commodity bundles $z$ in the commodity space such that $z + \omega(t) \succ_t x(t)$ and let $\Gamma^x$ be the convex hull of the union of all $\Gamma^x_t$, where the union is taken over any measurable set $U$ such that $\mu(T \setminus U) = 0$. As in \cite{Aumann64} $x$ is a competitive allocation if and only if the origin does not belong to $\Gamma^x$. Now, let $x$ be a core allocation. Since $x(t) \succeq_t \omega(t)$ or almost all $t$ we have $\Gamma^x_t \subseteq \Gamma^\omega_t$ and thus $\Gamma^x \subseteq \Gamma^\omega$. Since $\omega$ is a competitive allocation the origin does not belong to $\Gamma^\omega$, and thus also not to $\Gamma^x$. Hence $x$ is a competitive allocation. 
\end{proof}

\paragraph{Unbalanced economies} Two consumers are said to be of the same type if they have the same preference relation and the same initial endowment. An economy is called an \emph{n-type} economy if the set of consumers can be partitioned into $n$ non-null sets, each containing consumers of a single type. The economy is called \emph{unbalanced} if an element of this partition consists of an atom, i.e. there is an atom whose type is not shared by any other non-null set of agents.

\section{Results}

\subsection{Main results}

The main result is:

\begin{theorem}\label{non_competitive_core_allocation}
Let $\mathscr{E}$ be a two-type unbalanced non-trivial linear economy. There exists a non-competitive core allocation for $\mathscr{E}$.
\end{theorem}

Equivalently:

\begin{theorem}\label{atomless_measure_space} 
Let $(T, \mathscr{T}, \mu)$ be a measure space and $a_1$ and $a_2$ distinct vectors in $\mathbb{R}^l_{++}$ such $\sum_{i=1}^l a_ti = 1$ for $t = 1,2$. Consider the set of two-type non-trivial linear economies on the measure space $(T, \mathscr{T}, \mu)$ such that the preference relation of type $i$ consumers is represented by $u_i(x) = a_i \cdot x$. If for all such economies we have $C(\mathscr{E}) \subseteq W(\mathscr{E})$, then $(T, \mathscr{T}, \mu)$ has no atoms. 
\end{theorem}

The idea of the proof is to look at the core allocation most preferred by the atom whose type is unique and show that it cannot be competitive. The approach can be summarised as follows. Let $\mathscr{E}$ be a two-type unbalanced non-trivial linear economy. We first show that if one restricts attention to a certain class of allocations, then $\mathscr{E}$ is equivalent to a non-trivial linear economy with two atoms. We then show that for any non-trivial linear economy with two atoms there exists a core allocation that is not competitive.

\subsection{Rescaling lemma} Here we rehash a result from \cite{mertens2013game} that will allow us to rescale one economy to another in such a way that the relation between core and competitive equilibria remains fixed. The rescaling allows us to choose the measure of subsets of consumers. Later we will find it useful to rescale certain subsets of consumers to have measure one.

To define the transform, fix a set of consumers $T$. Let $\mathbb{E}$ denote the class of all economies in which the set of consumers is $T$ and $\Lambda$ the class of all simple functions $\lambda = \sum_{i=1}^n \lambda_i \chi_{G_i}$ (where $\chi_{G_i}$ denotes the indicator function on the set $G_i$) from the set of consumers to the positive real numbers. Define a mapping $\Xi: \mathbb{E} \times \Lambda \to \mathbb{E} $ as follows. For an economy $\mathscr{E} = ((T, \mathscr{T}, \mu), \omega, \succeq)$ in $\mathbb{E}$ and a simple function $\lambda$ in $\Lambda$ let $\Xi(\mathscr{E}, \lambda)$ be the economy $((T, \mathscr{T}, \mu), \hat \omega, \hat \succeq)$in $\mathbb{E}$ defined by:

\begin{enumerate}
  \item $\hat \omega = \lambda \omega$
  \item $\succeq_t$ is defined by $\lambda(t)x \hat \succeq \lambda(t)y$ if and only if $x \succeq_t y$
  \item $\hat \mu$ is defined by the equation
    \[\hat \mu(S) = \sum_{i=1}^n \frac{\mu(S \cap G_i)}{\lambda_i}\]
for any $S$ in $\mathscr{T}$.
\end{enumerate}  

Notice that $(\hat \mu)$ is a measure on $(T, \mathscr{T})$. Also, it is clear that if $x$ is an allocation and $\lambda$ an element of $\Lambda$, then $\lambda x$ is integrable. Thus the mapping is well defined. We now state and prove the rescaling lemma.

\begin{lemma}\label{rescaling_lemma}
Let $\mathscr{E}$ be an economy. For any $\lambda$ in $\Lambda$ we have:
\begin{enumerate}
  \item An allocation $x \in C(\mathscr{E})$ if and only if $\lambda x \in C(\mathscr{E})$
  \item An allocation $x \in W(\mathscr{E})$ if and only if $\lambda x \in W(\mathscr{E})$
\end{enumerate}
\end{lemma}

\begin{proof}
First it is shown that an allocation $x$ for $\mathcal{E}$ is feasible if and only if $\lambda x$ is feasible for $\Xi(\mathscr{E}, \lambda)$. To see this let $x$ be an allocation fo $\mathscr{E}$. Then 

\[
\int \lambda x d\hat\mu = \int \lambda x \frac{d\mu}{\lambda} = \int x d \mu.
\]

To see that point 1 is true let $x$ be a feasible allocation for $\mathscr{E}$. Suppose $x$ can be blocked by a coalition $S$. Then there exists an allocation $y$ for $\mathscr{E}$ such that $\int_S y= 'int_S \omega$ and $y(t) \succ_t x(t)$ for all $t$ in $S$. By the previous observation $\lambda y$ is a feasible allocation for $\Xi(\mathscr{E}, \lambda)$ and $\int_S \lambda y d \hat \mu = \int_S \lambda \omega d\hat\mu$. By the definition of $\hat \succeq$ we have $\lambda(t) y(t) \hat \succ_t \lambda(t) x(t)$ for all $t$ in $S$. Thus $\lambda x$ can be blocked by $S$ in $\Xi(\mathscr{E}, \lambda)$. The converse holds by symmetry.

To see that point 2 is true let $(p, x)$ be a competitive equilibrium for $\mathscr{E}$. Then for $\mu$-almost all consumers $t$ in $T$ we have $p \cdot x(t) \leq p \cdot \omega(t)$, and that $y \succ_t x(t)$ implies $p \cdot y > p \cdot \omega(t)$. Since $\lambda(t)$ is a positive real number $p \cdot \lambda(t) x(t) \leq p \cdot \lambda(t) \omega(t)$ $\mu$-almost everywhere. By the definition of $\hat \succeq$ it is clear that $\lambda(t) y \hat \succ_t \lambda(t) x(t)$ implies $p \cdot \lambda(t) y > p \cdot \lambda(t) \omega(t)$ $\hat\mu$-almost everywhere. Thus $(p, \lambda x)$ is a competitive equilibrium for $\Xi(\mathscr{E}, \lambda)$. Once again the converse holds by symmetry.
\end{proof}

\subsection{A two atom economy} This section studies the class of non-trivial linear economies with two atoms. The result of this section says that if the price system facing each atom is the vector defining the other atom's utility function then the solution to one of the resulting utility maximisation problems is infeasible with the resources of the economy.

\paragraph{Starting observations} Let $\mathscr{E}$ be a non-trivial linear economy with two atoms (such an economy is two-type and unbalanced). By Lemma \ref{rescaling_lemma} th    ere we can assume without loss of generality that the measure of each atom is one. Denote the atoms $A_1$ and $A_2$. Since the economy is linear each atom's preferenc relation can be represented by a linear utility function $u_i: \mathbb{R}^l_+ \rightarrow \mathbb{R}$ defined by $u_i(x) = a_i \cdot x$ where $a_i$ belongs to $\mathbb{R}^l_+$ for both $i=1, 2$ such that $\sum_{j=1}^l a_{ij}$. Since the economy is non-trivial $a_i \in \mathbb{R}^l_{++}$ and $a_1 \neq a_2$ (for otherwise $\omega$ would be Pareto optimal.) Fix a price system $p$ and consider the solution to the following utility maximisation problem 

\[
x^*_i(p) = \arg \max \{a_i \cdot x_i : x_i \in \mathbb{R}^l_+, p\cdot x \leq p \cdot \omega(A_i)\}.
\] 

Let $x$ be a commodity bundle. Define the \emph{support} of $x$ to be $\text{supp} x = \{l:x_l > 0\}$. For each $A_i$, $i=1,2$, define the set of commodities yielding maximal marginal utility by $S(p, a_i) = \{i :  a_ip_k \geq a_kp_i \text{ for all } k=1,\cdots,l\}$. It is clear that each atom consumes only those commodities yielding maximal marginal utility.

\begin{lemma}\label{subset}
$\text{supp } x^*_i(p) \subseteq S(p, a_i)$
\end{lemma}

We now show that if each atom $A_i$ faces prices $p = a_j$ with $j \neq i$, then the sets of commodities that yield maximal marginal utility for each agent are disjoint.

\begin{lemma}\label{disjoint}
$S(a_2, a_1) \cap S(a_1, a_2) = \emptyset$
\end{lemma}

\begin{proof}
Suppose not. Then there exists a commodity $i$ belonging to both $S(a_2, a_1)$ and $S(a_1, a_2)$ so that $a_{1i} a_{2k} \geq a_{1k}a_{2i}$ and $a_{2i} a_{1k} \geq a_{2k}a_{1i}$ for all $k=1,\cdots,l$. This implies

\[
a_{1i}a_{2k} = a_{1k}a_{2i} \text{ for all } k = 1, \cdots, l.
\]

If $a_{1i} = a_{2i}$, then $a_1 = a_2$, a contradiction. Thus there exists a positive number $\alpha$, not equal to one, such that $a_{1i} = \alpha a_{2i}$. But then $a_{1k} = \alpha a_{2k}$ for all $k=1,\cdots, l$ so that 

\[1 = \sum_{k = 1} ^ l a_{1k} = \sum_{k=1}^l \alpha a_{2k} = \alpha,\]
a contradiction.
\end{proof}

\paragraph{Decentralization}   

We now show that price systems $p = a_1$ and $p = a_2$ are not part of any competitive equilibrium.

\begin{lemma}\label{no_decentralization}
$x^*_1(a_2) \not\leq \omega(A_1) + \omega(A_2)$ and $x^*_2(a_1) \not\leq \omega(A_1) + \omega(A_2)$
\end{lemma}

\begin{proof}
Suppose not. Then $x^*_1(a_2) \leq \omega(A_1) + \omega(A_2)$ or $x^*_2(a_1) \leq \omega(A_1) + \omega(A_2)$. By Lemma \ref{subset} and Lemma \ref{disjoint} 

\[
\text{supp } x^*_1(a_2) \cap \text{ supp } x^*_2(a_1) = \emptyset. 
\]

This together with $x^*_1(a_2) \leq \omega(A_1) + \omega(A_2)$ and $x^*_2(a_1) \leq \omega(A_1) + \omega(A_2)$ gives

\[
x^*_1(a_2) + x^*_2(a_1) \leq \omega(A_1) + \omega(A_2).
\]

By Walras' law $a_1 \cdot x^*_2(a_1) = a_1 \cdot \omega(A_2)$ which, using the previous inequality, gives $a_1 \cdot x^*_1(a_2) \leq a_1 \cdot \omega(A_1)$. This turns out to be a contradiction. 

To see this recall that $\omega$ is not Pareto optimal so there exist commodity bundles $x_1$ and $x_2$ such that $x_1 + x_2 = \omega(A_1) + \omega(A_2)$, $a_1 \cdot x_1 > a_1 \cdot \omega(A_1)$ and $a_2 \cdot x_2 > a_2 \cdot \omega(A_2)$. Commodity bundle $x_1$ is affordable since $a_2 \cdot (x_1 + x_2) = a_2 \cdot( \omega(A_1) + \omega(A_2))$ and the last inequality implies

\[
a_1 \cdot x_1 < a_2 \cdot \omega(A_1).
\]

Since $x_1$ is affordable under prices given by $a_2$ and it is also preferred to the initial allocation $\omega(A_1)$ it must be that $x_1^*(a_2)$ is also preferred to the initial allocation $\omega(A_1)$, i.e. $a_1 \cdot x^*_1(a_2) > a_1 \cdot \omega(A_1)$ providing the desired contradiction.
\end{proof}

\paragraph{Core allocation} 

This section studies the class of two-type non-trivial linear economies. The main result is a characterisation of a core allocation for such economies.

Fix a two-type unbalances non-trivial economy $\mathscr{E}$. Since $\mathscr{E}$ is unbalanced there is an atom of a unique type. Denote this atom $A_1$. We will call consumers belonging to $A_1$ type one and consumers belonging to $T \setminus A_1$ type two. To simplify notation we denote the initial endowment of a consumer of type $i$ by $\omega_i$. Denote the utility function representing the preferences of a consumer of type $i$ by $u_i(x) = a_i \cdot x$.

It happens that in such aneconomy there is a core allocation in which all consumers of type two get a commodity bundle indifferent to their initial commodity bundle and the atom of type one maximises his utility over the set of feasible allocations satisfying this constraint. We now construct this allocation. Consider the following set of allocations $X = F \cap W$ where 

\[
F = \{(x_1, x_2) \in \mathbb{R}^{2l}_+ : \mu(A_1)x_1 + \mu(T \setminus A_1)x_2 \leq \mu(A_1) \omega_1 + \mu(T \setminus A_1) \omega_2\}
\]

\[W = \{(x_1, x_2) \in \mathbb{R}^{2l}_+: u_2(x_2) = u_2(\omega_2)\}.\]

Define the maximiser

\begin{equation}\label{core_max}
(x^*_1, x^*_2) = \arg \max \{u_1(x_1) : (x_1, x_2) \in X\}.
\end{equation}

This is well defined: X is closed and bounded. By the Heine-Borel property X is compact. As $u_1$ is continuous we obtain a solution by Weierstrass' theorem.

The main result of this section is:

\begin{lemma}\label{core_allocation}
The allocation $x^*:T \to \mathbb{R}^l_+$ that assigns $x_i^*$ defined in equation \ref{core_max} to each consumer of type $i$ belongs to $C(\mathscr{E})$.
\end{lemma}  

\begin{proof}
Clearly $x^*$ is a feasible allocation. Suppose by way of contradiction that $x^*$ does not belong to $C(\mathscr{E})$. Then there exists a coalition $S$ which blocks $x^*$ with an allocation $y$. Note that $A_1$ belongs to $S$. For if not, then by the linearity of $u_2$ we have that $(1/ \mu(S)) \int_S y$ is preferred to $x_2^*$ by each consumer of type two. However $(1 / \mu(S)) \int_S y = \omega_2$ contradicting $u_2(x_2^*) = u_2(\omega_2)$. By the same argument a non-null set of consumers of type two belongs to $S$. Define

\[\bar y = \frac{1}{\mu(T \setminus S)}\left[ \int_{S \setminus A_1} y + \int_{T \setminus S} \omega_2\right].\]
Clearly $\bar y$ is preferred to $\omega_2$ by each consumer of type two. Since $u_2$ is continuous, by the Intermediate Value Theorem there exists an $\alpha \in (0, 1)$ such that $\alpha \bar y$ is indifferent to $\omega_2$ for each consumer of type two. Finally,

\begin{eqnarray*}
\mu(A_1) y(A_1) + \mu(T \setminus A_1) \alpha \bar y &\leq& \mu(A_1) y(A_1) + \mu(T \setminus A_1) \bar y\\
&=& \mu(A_1) y(A_1) + \int_{S \setminus A_1} y + \int_{T \setminus S} \omega_2\\
&=& \int \omega\\
&=& \mu(A_1) \omega_1 + \mu(T \setminus A_1) \omega_2. 
\end{eqnarray*}

Thus $(y(A_1), \alpha \bar y)$ belongs to $X$ contradicting that $(x^*_1, x^*_2)$ is maximal.
\end{proof}

\paragraph{Proof of the main theorem}

In any two-type unbalanced non-trivial linear economy Lemma \ref{core_allocation} characterises a core allocation. It turns out via Lemma \ref{no_decentralization} that this allocation is not competitive. Thus in every two-type unbalanced non-trivial linear economy there exists a non-competitive core allocation.

\begin{proof}[Proof of Theorem \ref{non_competitive_core_allocation}]
By Lemma \ref{rescaling_lemma} there is no loss of generality in assuming that $\mu(A_1) = \mu(T \setminus A_1) = 1$. Since $\mathscr{E}$ is a two type unbalanced economy Lemma \ref{core_allocation} yields the core allocation $x^*$. Suppose by contradition that $x^*$ is a competitive allocation with price system $p$. We claim that $p = a_2$. For if $p \neq a_2$, then since $\omega_2$ is not Pareto optimal there exists a commodity bundle $y$ such that $p \cdot y \leq p \cdot \omega_2$ and $a_2 \cdot y > a_2 \cdot \omega_2$. Since $a_2 \cdot x_2^* = a_2 \cdot \omega_2$ this contradicts our assumption that $x^*$ is a competitive allocation. Thus $p = a_2$. By Lemma \ref{no_decentralization}, interchanging $a_1$ and $a_2$ if necessary, we have $x_2 \not \leq \omega_1 + \omega_2$. Thus $\int x^* \not \leq \int \omega.$ A contradiction.
\end{proof}

Finally,

\begin{proof}[Proof of Theorem \ref{atomless_measure_space}]
By Theorem \ref{non_competitive_core_allocation} if $(T, \mathscr{T}, \mu)$ is atomic, then $C(\mathscr{E}) \not\subseteq W(\mathscr{E})$ for some such two-type non-trivial linear economy.
\end{proof}

\newpage


\end{document}